\documentclass[12pt]{article}
\usepackage{amsmath, amsthm}
\usepackage{amscd}
\usepackage{amssymb}
\usepackage{array}
\usepackage{color}
\usepackage{youngtab}

\newtheorem{thm}{Theorem}[section]

\newtheorem{corollary}[thm]{Corollary}

\newtheorem{proposition}[thm]{Proposition}

\theoremstyle{definition}
\newtheorem{definition}[thm]{Definition}
\newtheorem{example}[thm]{Example}

\newtheorem{remark}[thm]{Remark}

\begin{document}

\newcommand{\comment}[1]{{\color{blue}\rule[-0.5ex]{2pt}{2.5ex}}
\marginpar{\scriptsize\begin{flushleft}\color{blue}#1\end{flushleft}}}

\newcommand{\be}{\begin{equation}}
\newcommand{\ee}{\end{equation}}
\newcommand{\bea}{\begin{eqnarray}}
\newcommand{\eea}{\end{eqnarray}}
\newcommand{\bean}{\begin{eqnarray*}}
\newcommand{\eean}{\end{eqnarray*}}

\newcommand{\id}{\relax{\rm 1\kern-.28em 1}}
\newcommand{\R}{\mathbb{R}}
\newcommand{\C}{\mathbb{C}}
\newcommand{\Z}{\mathbb{Z}}
\newcommand{\Q}{\mathbb{Q}}
\newcommand{\g}{\mathfrak{G}}
\newcommand{\e}{\epsilon}

\newcommand{\hs}{\hfill\square}
\newcommand{\hbs}{\hfill\blacksquare}

\newcommand{\bp}{\mathbf{p}}
\newcommand{\bmax}{\mathbf{m}}
\newcommand{\bT}{\mathbf{T}}
\newcommand{\bU}{\mathbf{U}}
\newcommand{\bP}{\mathbf{P}}
\newcommand{\bA}{\mathbf{A}}
\newcommand{\bm}{\mathbf{m}}
\newcommand{\bIP}{\mathbf{I_P}}

\newcommand{\cA}{\mathcal{A}}
\newcommand{\cB}{\mathcal{B}}
\newcommand{\cC}{\mathcal{C}}
\newcommand{\cI}{\mathcal{I}}
\newcommand{\cO}{\mathcal{O}}
\newcommand{\cG}{\mathcal{G}}
\newcommand{\cJ}{\mathcal{J}}
\newcommand{\cF}{\mathcal{F}}
\newcommand{\cP}{\mathcal{P}}
\newcommand{\ep}{\mathcal{E}}
\newcommand{\E}{\mathcal{E}}
\newcommand{\cH}{\mathcal{O}}
\newcommand{\cPO}{\mathcal{PO}}
\newcommand{\cl}{\ell}
\newcommand{\cFG}{\mathcal{F}_{\mathrm{G}}}
\newcommand{\cHG}{\mathcal{H}_{\mathrm{G}}}
\newcommand{\Gal}{G_{\mathrm{al}}}
\newcommand{\cQ}{G_{\mathcal{Q}}}

\newcommand{\ri}{\mathrm{i}}
\newcommand{\re}{\mathrm{e}}
\newcommand{\rd}{\mathrm{d}}

\newcommand{\rGL}{\mathrm{GL}}
\newcommand{\rSU}{\mathrm{SU}}
\newcommand{\rSL}{\mathrm{SL}}
\newcommand{\rSO}{\mathrm{SO}}
\newcommand{\rOSp}{\mathrm{OSp}}
\newcommand{\rSpin}{\mathrm{Spin}}
\newcommand{\rsl}{\mathrm{sl}}
\newcommand{\rM}{\mathrm{M}}
\newcommand{\rdiag}{\mathrm{diag}}
\newcommand{\rP}{\mathrm{P}}
\newcommand{\rdeg}{\mathrm{deg}}

\newcommand{\M}{\mathrm{M}}
\newcommand{\End}{\mathrm{End}}
\newcommand{\Hom}{\mathrm{Hom}}
\newcommand{\diag}{\mathrm{diag}}
\newcommand{\rspan}{\mathrm{span}}
\newcommand{\rank}{\mathrm{rank}}
\newcommand{\Gr}{\mathrm{Gr}}
\newcommand{\ber}{\mathrm{Ber}}

\newcommand{\fsl}{\mathfrak{sl}}
\newcommand{\fg}{\mathfrak{g}}
\newcommand{\ff}{\mathfrak{f}}
\newcommand{\fgl}{\mathfrak{gl}}
\newcommand{\fosp}{\mathfrak{osp}}
\newcommand{\fm}{\mathfrak{m}}

\newcommand{\str}{\mathrm{str}}
\newcommand{\Sym}{\mathrm{Sym}}
\newcommand{\tr}{\mathrm{tr}}
\newcommand{\defi}{\mathrm{def}}
\newcommand{\Ber}{\mathrm{Ber}}
\newcommand{\spec}{\mathrm{Spec}}
\newcommand{\sschemes}{\mathrm{(sschemes)}}
\newcommand{\sschemeaff}{\mathrm{ {( {sschemes}_{\mathrm{aff}} )} }}
\newcommand{\rings}{\mathrm{(rings)}}
\newcommand{\Top}{\mathrm{Top}}
\newcommand{\sarf}{ \mathrm{ {( {salg}_{rf} )} }}
\newcommand{\arf}{\mathrm{ {( {alg}_{rf} )} }}
\newcommand{\odd}{\mathrm{odd}}
\newcommand{\alg}{\mathrm{(alg)}}
\newcommand{\sa}{\mathrm{(salg)}}
\newcommand{\sets}{\mathrm{(sets)}}
\newcommand{\SA}{\mathrm{(salg)}}
\newcommand{\salg}{\mathrm{(salg)}}
\newcommand{\varaff}{ \mathrm{ {( {var}_{\mathrm{aff}} )} } }
\newcommand{\svaraff}{\mathrm{ {( {svar}_{\mathrm{aff}} )}  }}
\newcommand{\ad}{\mathrm{ad}}
\newcommand{\Ad}{\mathrm{Ad}}
\newcommand{\pol}{\mathrm{Pol}}
\newcommand{\Lie}{\mathrm{Lie}}
\newcommand{\Proj}{\mathrm{Proj}}
\newcommand{\rGr}{\mathrm{Gr}}
\newcommand{\rFl}{\mathrm{Fl}}
\newcommand{\rPol}{\mathrm{Pol}}
\newcommand{\rdef}{\mathrm{def}}

\newcommand{\uspec}{\underline{\mathrm{Spec}}}
\newcommand{\uproj}{\mathrm{\underline{Proj}}}

\newcommand{\sym}{\cong}

\newcommand{\al}{\alpha}
\newcommand{\lam}{\lambda}
\newcommand{\de}{\delta}
\newcommand{\ttau}{\tilde \tau}
\newcommand{\D}{\Delta}
\newcommand{\s}{\sigma}
\newcommand{\lra}{\longrightarrow}
\newcommand{\ga}{\gamma}
\newcommand{\ra}{\rightarrow}

\newcommand{\NOTE}{\bigskip\hrule\medskip}


\smallskip

\medskip

 \centerline{\LARGE \bf On Chiral Quantum Superspaces}

\vskip 1cm

\centerline{ D. Cervantes.}

\smallskip

\centerline{\it Instituto de Ciencias Nucleares,
Universidad Nacional Aut\'{o}noma de M\'{e}xico }
 \centerline{\it Circuito Exterior M\'{e}xico D.F. 04510, M\'{e}xico}
\centerline{{\footnotesize email: daliac@nucleares.unam.mx}}

\bigskip

\centerline{ R. Fioresi
}

\smallskip

\centerline{\it Dipartimento di Matematica, Universit\`{a} di
Bologna }
 \centerline{\it Piazza di Porta S. Donato, 5. 40126 Bologna. Italy.}
\centerline{{\footnotesize e-mail: fioresi@dm.UniBo.it}}

\bigskip

\centerline{ M. A. Lled\'{o} }

\smallskip

 \centerline{\it  Departamento de F\'{\i}sica Te\`{o}rica,
Universitat de Val\`{e}ncia and IFIC (CSIC-UVEG)}
\centerline{\it  Fundaci\'{o} General Universitat de Val\`{e}ncia.}
 \centerline{\small\it C/Dr.
Moliner, 50, E-46100 Burjassot (Val\`{e}ncia), Spain.}
 \centerline{{\footnotesize e-mail: maria.lledo@ific.uv.es}}

\vskip 1cm

\begin{abstract}
We give a quantum deformation
of the  chiral Minkowski superspace in 4 dimensions embedded as
the big cell into
the chiral conformal superspace. Both deformations are
realized as quantum homogeneous superspaces:
we deform the ring of regular functions
together with a coaction of the corresponding quantum supergroup.
\end{abstract}

\section{Introduction}

In his foundational work on supergeometry \cite{ma1} Manin
realized the Minkowski superspace as the big
cell inside the flag supermanifold of $2|0$ and $2|1$ superspaces in the
superspace of dimension $4|1$.

\medskip

In his construction however, the actions of the Poincar\'{e}
and the conformal supergroups on the super Minkowski and its
compactification were left in the background
and did not play a crucial role.
Moreover there was no explicit construction of the coordinate
rings associated with the super Minkowski space and the conformal
superspace together with their embedding into a suitable projective
superspace. Such coordinate rings are necessary in order to construct
a quantum deformation.

\medskip

Our intention is to fill this gap, by bringing the supergroup action
to the center of the stage so that we can give explicitly the coordinate
rings of the super Minkowski and
conformal superspaces together with their embeddings
into projective superspace.
This will be our starting point to build a quantum deformation of them.
We shall concentrate our
attention in realizing the chiral super Minkowski space as the big cell in the
super Grassmannian variety
of $2|0$ superspaces in $\C^{4|1}$ (the chiral conformal superspace). This is not
precisely the same supervariety that Manin considers in his work; the Grassmannian is a simpler one, but it also has a physical meaning. Our choice is motivated because in some supersymmetric theories {\it chiral superfields} appear naturally. Chiral superfields, in our approach, are  identified with elements of the coordinate superalgebra of the above mentioned Grassmannian. If one wants to formulate certain supersymmetric field  theories in a noncommutative superspace one
needs to have the notion of quantum chiral superfields. It is not obvious in other approaches how to construct a quantum chiral superalgebra without loosing other properties, as the action of the group, for example. In our construction the quantum chiral superfields appear naturally together with the supergroup action.

 We plan to explore in a forthcoming paper Manin's
construction in this new framework.

\medskip

We shall not go into the details of the proofs of all of our statements,
since an enlarged version of part of this work is available in
Ref. \cite{cfl}; neverthless we shall make a constant effort to convey the
key ideas and steps of our constructions.

\medskip

This is the content of the present paper.

\medskip

In section \ref{supergeo} we briefly outline few key facts
of supergeometry, favouring intuition over rigorous definitions.
Our main reference will be Ref. \cite{cf}.

\medskip

In section \ref{conformal} we discuss the chiral conformal superspace
as an homogeneous superspace identified with the super Grassmannian
variety of $2|0$ superspaces in the complex vector superspace of dimension
$4|1$. We also provide an explicit projective
embedding of the super Grassmannian into a suitable projective superspace.

\medskip

In Section \ref{invar} we give an equivalent approach via invariant theory
to the theory discussed in Section \ref{conformal}.

\medskip

In Section \ref{smink} we introduce the complex super Minkowski space as the big
cell in the chiral conformal superspace.
We also provide an explicit description
of the action of the super Poincar\'{e} group.

\medskip

In Sections \ref{qconformal} and
\ref{qmink} we build a quantum deformation of the Minkowski
superspace and its compactification together with a coaction of the
quantum Poincar\'{e} and conformal supergroups.

\medskip

Finally in Section \ref{chiral} we discuss
some relevant physical applications of the theory
developed so far.

\medskip

{\bf Acknoledgements}. The authors wish to thank the UCLA Department
of Mathematics for the wonderful hospitality during the workshop, that
made the present work possible. The authors wish also to thank prof.
V. S. Varadarajan for the many helpful discussions on supergeometry
and supergroups.

\section{Basic concepts in Supergeometry}
\label{supergeo}

Supergeometry is essentially $\Z_2$-graded geometry: any geometrical
object is given a $\Z_2$-grading in some natural way and the morphisms
are the maps respecting the geometric structure and the $\Z_2$-grading.

\medskip

For instance, a super vector space $V$ is a vector space where we
establish a $\Z_2$-grading by giving a splitting $V_0 \oplus V_1$.
The elements in $V_0$ are called {\it even}
and the elements in $V_1$ are called {\it odd}.
Hence we have a function $p$
called the \textit{parity} defined only on homogeneous elements.
A \textit{superalgebra} $A$ is a super vector space with multiplication
preserving parity. The \textit{reduced superalgebra} associated with $A$
is $A_r:=A/I_{\mathrm{odd}}$, where $I_{\mathrm{odd}}$ is the ideal generated by the
odd nilpotents. Notice that the reduced superalgebra $A_r$ may have
even nilpotents, thus making the terminology a bit awkward.

\medskip

A superalgebra $A$ is \textit{commutative} if
$$
xy=(-1)^{p(x)p(y)}yx
$$
for all
$x$, $y$ homogeneous elements in $A$.
>From now on we assume all superalgebras are to be
commutative unless otherwise specified
and their category is
denoted with $\salg$.
We also need to introduce the notion of \textit{affine superalgebra}.
This is a finitely generated superalgebra such that $A_r$ has no
nilpotents. In ordinary algebraic geometry such $A_r$'s are
associated bijectively to affine algebraic varieties, as we
are going to see.

\medskip

The most interesting objects in supergeometry are the
\textit{algebraic supervarieties} and the
\textit{differentiable supermanifolds}. Both these concepts
are encompassed by the idea of \textit{superspace}.

\begin{definition}
We define \textit{superspace} the pair
$S=(|S|, \cO_S)$ where $|S|$ is a topological space and $\cO_S$ is
a sheaf of superalgebras such that the stalk at
a point  $x\in |S|$ denoted by $\cO_{S,x}$ is a local superalgebra
for all $x \in |S|$.

A {\it morphism} $\phi:S \lra T$ of superspaces is given by
$\phi=(|\phi|, \phi^\#)$, where $\phi: |S| \lra |T|$ is a map of
topological spaces and $\phi^\#:\cO_T \lra \phi_*\cO_S$ is
a sheaf morphism such that
$\phi_x^\#(\bm_{|\phi|(x)})=\bm_x$ where $\bm_{|\phi|(x)}$ and
$\bm_{x}$ are the maximal ideals in the stalks $\cO_{T,|\phi|(x)}$
and $\cO_{S,x}$ respectively.
\end{definition}

Let us see an important example.

\begin{example}
The superspace $\R^{p|q}$ is the topological space $\R^p$ endowed with
the following sheaf of superalgebras. For any $U
\subset_\mathrm{open} \R^p$
$$
\cO_{\R^{p|q}}(U)=C^{\infty}(\R^p)(U)\otimes \R[\xi^1, \dots, \xi^q],
$$
where $\R[\xi_1, \dots ,\xi_q]$ is the exterior algebra (or {\it
Grassmann algebra}) generated by the $q$ variables $\xi_1 ,\dots,
\xi_q$.
\end{example}

\begin{definition}\label{supermanifold}
A  \textit{supermanifold} of dimension $p|q$ is a superspace $M=(|M|, \cO_M)$
which is locally isomorphic to the superspace
$\R^{p|q}$, i. e. for all $x \in |M|$
there exist an open set $V_x \subset |M|$ and  $U \subset \R^{p|q}$
such that:
$$
{\cO_{M}}|_{V_x} \cong {\cO_{\R^{p|q}}}|_U.
$$
\end{definition}

We shall now concentrate on the study of algebraic supervarieties,
since our purpose is
to obtain quantum deformations and for this reason the algebraic
approach is to be preferred.

\medskip

There are two equivalent and quite different approaches to
both, algebraic supervarieties and differentiable supermanifolds:
the sheaf theoretic and the functor of points categorical approach.
In the first of these approaches an algebraic
supervariety (resp. a supermanifold) is to be understood
as a superspace, that is a
pair consisting of a topological space and a sheaf of
superalgebras. In the special cases of
an affine algebraic supervariety (resp. a differentiable
supermanifold), the superalgebra of global
sections of the sheaf allows us to reconstruct the whole sheaf
and the underlying topological space (see \cite{cf} ch. 4 and 10).
Consequently an affine supervariety (resp. a differentiable
supermanifold) can be effectively
identified with a commutative superalgebra.

\medskip

This is the super counterpart to the well known result of
ordinary complex algebraic geometry:
affine varieties are in one-to-one correspondence
with their coordinate rings, in other words,  we associate
the zeros of a set of polynomials into some affine space to the ideal
generated by such polynomials. For example we associate to the
complex sphere in $\C^3$, the coordinate ring $\C[x,y,z]/(x^2+y^2+z^2-1)$.

We also say that there is an equivalence of
categories between the category of affine supervarieties and the category
of affine superalgebras.
Besides the above mentioned correspondence, this amounts to the fact
that morphisms of affine
varieties correspond to morphisms of the correspondent
coordinate rings.

\medskip

We can take the same point of view in supergeometry and
give the following definition.

\begin{definition}
Let $\cO(X)$ be an affine superalgebra.
We define \textit{affine supervariety} $X$ associated with $\cO(X)$
the superspace $(|X|, \cO_X)$, where $|X|$ is
the topological space of an ordinary affine variety, while
$\cO_X$ is the (unique) sheaf of superalgebras, whose global
sections coincide with $\cO(X)$,
and there exists an open cover $U_i$ of $|X|$ such that
$$
\cO_X(U_i)=\cO(X)_{f_i}=\left\{{g \over f_i} \, \big| \, g \in \cO(X) \right\}
$$
for suitable $f_i \in \cO(X)_0$.
(for more details see \cite{eh} ch. II and \cite{cf} ch. 10).

A \textit{morphism} of affine supervarieties is a morphism of
the underlying superspaces, though one readily see it corresponds
(contravariantly) to a morphism of the corresponding coordinate superalgebras:
$$
\hbox{\{morphisms $X \lra Y$ \} }
\qquad \longleftrightarrow  \qquad
\hbox{ \{morphisms $\cO(Y) \lra \cO(X)$ \} }
$$

We define \textit{algebraic supervariety} a superspace
which is locally isomorphic to an affine supervariety. $\hfill\Box$
\end{definition}

\medskip

\begin{example} \label{affinespace} ${}^{}$

{\sl 1. The affine superspace.}
We define the
\textit{polynomial superalgebra} as:
$$
\C[x^1,\dots ,x^p,\theta^1,\dots ,\theta^q]
:= \C[x^1,\dots ,x^p]\otimes \Lambda (\theta^1,\dots ,\theta^q).
$$
We want to interpret this superalgebra as
the coordinate superring of a supervariety that we call
the \textit{affine superspace} of superdimension $p|q$,
and we shall denote with the symbol $\C^{p|q}$ or $\bA^{m|n}$.
The underlying topological space is $\bA^m$, that
is $\C^m$ with the Zariski topology, while the sheaf is:
$$
\cO_{\bA^{m|n}}(U):=\cO_{\bA^m}(U) \otimes \Lambda(\theta^1 \dots \theta^n).
$$

\medskip

{\sl 2. The supersphere.}
The superalgebra $\C[x_1,x_2,x_3]/(x_1^2+x_2^2+x_3^2+
\eta_1 x + \eta_2 x_2 +\eta_3 x_3 - 1)$ is the superalgebra of the
global sections of an affine supervariety whose underlying topological
space is the unitary sphere in $\bA^3$.

\end{example}

The first important example of a supervariety which is not affine
is given by the \textit{projective superspace}.

\begin{example} ${}^{}$

\label{projsup} {\sl 1. Projective superspace.}
Consider the $\Z$-graded superalgebra $S=\C[x_0 \dots x_{m}, \xi_1 \dots
\xi_n]$.
For each $r$, $0 \leq r \leq  m$, we consider the graded superalgebra
$$
S[r]=\C[x_0,\dots ,x_m, \xi_1,\dots, \xi_n][x_r^{-1}],\qquad
\deg(x_r^{-1})=-1.
$$
The subalgebra $S[r]^0\subset S[r]$ of $\Z$-degree 0 is
\begin{equation}
S[r]^0\approx \C[u_0,\dots,\hat u_r ,\cdots, u_m,\eta_1,\dots
\eta_n],\qquad u_s=\frac {x_s}{x_r},\; \eta_\alpha=\frac
{\xi_\alpha}{x_r}, \label{homogeneous}
\end{equation}
(the `$\;\,\hat{}\;\,$' means that this generator is omitted).
This is an affine superalgebra and it corresponds to an affine superspace,
(see \ref{affinespace}) whose topological space we denote with
$|U_r|$ and the corresponding sheaf with $\cO_{U_r}$.
Notice that the topological spaces $|U_r|$ form an affine open
cover of $|\bP^m|$, the ordinary projective space of
dimension $m$.

A direct calculations shows that:
$$
\cO_{U_r}|_{|U_r| \cap |U_s|}=
\cO_{U_s}|_{|U_r| \cap |U_s|},
$$
so we conclude that there exists a unique sheaf on
the topological space $|\bP^m|$, that we
denote as  $\cO_{\bP^{m|n}}$, whose restriction to $|U_i|$ is
$\cO_{U_i}$. Hence we have defined a supervariety that we
denote with $\bP^{m|n}$ and call the \textit{
projective superspace} of dimension $m|n$.

\medskip

\noindent
{\sl 2. Projective supervarieties.}

\noindent
Let $I\subset S=\C[x_1 \dots x_m, \xi_1 \dots \xi_n]$ be a
homogeneous ideal; then $S/I$ is also a graded superalgebra and we
can repeat the same construction as above.
First of all, we notice that the reduced algebra $(S/I)_r$ corresponds
to an ordinary projective variety, whose topological space we denote
with $|X|$, embedded into a projective superspace $|X| \subset |\bP^m|$.
Consider the
superalgebra of $\Z$-degree zero elements in $(S/I)[x_i^{-1}]$
(this is called \textit{projective localization}):
\begin{eqnarray*}
\left(\frac {\C[x_0,\dots x_m,\xi_1\dots \xi_n]}{I}[x_i^{-1}]\right)_0 \cong
\frac {\C[u_0,\dots,\hat u_i,\dots u_m,\eta_1\dots
\eta_n]}{I_{\mathrm{loc}}},
\end{eqnarray*}
where $I_{\mathrm{loc}}$ are the even
elements of $\Z$-degree zero in $I[x_i^{-1}]$.

Again this affine superalgebra defines an affine supervariety with
topological space $|V_i|
\subset |U_i| \subset |\bP^m|$ and sheaf $\cO_{V_i}$.
One can check that the supersheaves $\cO_{V_i}$ are such that
$\cO_{V_i}|_{|V_i| \cap |V_j|}=\cO_{V_j}|_{|V_i| \cap |V_j|}$,
so they glue to give a sheaf on $|X|$. Hence
as before there exists a supervariety corresponding to the
homogeneous superring $S/I$. This supervariety comes equipped
with a projective embedding, encoded by the morphism of graded
superalgebra $S \lra S/I$, hence $(|X|, \cO_X)$ is called
a \textit{projective supervariety}. $\hfill\Box$
\end{example}

It is very important to remark that, contrary
to the affine case, there is no coordinate
superring associated instrinsecally to a projective supervariety,
but there is a coordinate superring associated with the projective
supervariety and its projective embedding. In other words we can
have the same projective variety admitting non isomorphic
coordinate superrings with respect to two different projective
embeddings.

\medskip

We now want to introduce the functor of points approach
to the theory of supervarieties.

Classically we can examine the points of a variety over
different fields and rings. For example we can look at the rational
points of the complex sphere described above. They are in
one to one correspondence with the morphisms:
$\C[x,y,z]/(x^2+y^2+z^2-1) \lra \Q$. In fact each such morphism
is specified by the knowledge of the images
of the generators.
The idea behind the functor of points is to extend this and consider
{\sl all} morphisms from the coordinate ring of
the affine supervariety to {\sl all} superalgebras at once.

\begin{definition}
Let $\cA \in \salg$, the category of commutative superalgebras.
We define the $\cA$-points of an affine supervariety $X$
as the (superalgebra) morphisms $\mathrm{Hom}(\cO(X), \cA)$.
We define the functor of points of $X$ as:
$$
h_X:\salg \lra \sets, \qquad h_X(\cA)={\mathrm Hom}(\cO(X), \cA).
$$
In other words $h_X(\cA)$ are the $\cA$-points of $X$,
for all commutative superalgebras $\cA$.
\end{definition}

\begin{example} \label{affinefopts}
If $\cA$ is a generic (commutative) superalgebra, an $\cA$-point of
$\C^{p|q}$ (see Example \ref{affinespace}) is given by a morphism
$\C[x^1,\dots ,x^p,\theta^1,\dots ,\theta^q] \lra \cA$, which is
determined once we know the image of the generators
$$(x^1,\dots ,x^p,\theta^1,\dots ,\theta^q)
\lra(a^1 ,\dots, a^p,\al^1, \dots, \al^q),$$
with $a^i \in \cA_0$ and $\al^j \in \cA_1$. Notice that the $\C$-points
of $\C^{p|q}$ are given by $(k_1 \dots k_p, 0 \dots 0)$ and coincide
with the points of the affine space $\C^p$. In this example it is
clear that the knowledge of the points over a field is by no means
sufficient to describe the supergeometric object.
$ $
\end{example}

\begin{remark}
It is important at this point to notice that just giving a
functor from $\salg$ to $\sets$, does not guarantee that it is the
functor of points of a supervariety. A set of conditions to
establish this is given in \cite{cf} ch. 10.
\end{remark}

The functor of points for projective supervarieties is more
complicated and we are unable to give a complete discussion here.
it would be too long to give a general discussion here.
We shall neverthless discuss the functor of points of the
projective space and superspace.

\begin{example} \label{projectivespace}
Let us consider the functor:
$h: \alg \lra \sets$, where $h(\cA)$ are the projective $\cA$-modules of
rank one in $\cA^n$.

Equivalently $h(\cA)$ consists
of the pairs $ (L,\phi)$, where $L$ is a projective
$\cA$-module of rank one, and $\phi$ is a surjective morphisms $\phi: \cA^{n+1} \lra L$. These pairs are taken modulo the equivalence relation
$$
(L,\phi)\approx(L',\phi')\qquad \Leftrightarrow\qquad
L \stackrel{a} \approx L',\qquad \phi'=a\circ\phi,
$$

 If $\cA=\C$, then projective modules are free and a morphism
$$\phi :\C^{n+1}\rightarrow \C$$
is specified by a n-tuple, $(a^1,\dots a^{n+1})$, with $a^i\in \C$,
not all of the $a^i=0$. The equivalence relation becomes
$$(a^1,\dots ,a^{n+1})\sim(b^1,\dots b^{n+1}) \quad \Leftrightarrow\quad (a^1,\dots, a^{n+1})=\lambda(b^1,\dots, b^{n+1}),$$
with $\lambda\in \C^\times$ understood as an automorphism of $\C$. It is clear
then that
$h(\C)$ consists of all the lines through the origin
in the vector space $\C^{n+1}$, thus
recovering the usual definition of complex projective space.

If $\cA$ is local, projective modules are free over local
rings. We then have a situation similar to the field setting:
equivalence classes are lines in the $\cA$-module $\cA^{n+1}$.

Using the Representability Theorem (see \cite{cf}) one can show that the
functor $h$ is the functor
of points of a variety that we call the projective space and whose
geometric points  coincide with the projective
space $\bP^n$ over the field $k$ as we usually understand it. $\hfill\Box$

\medskip

This example can be easily generalized to the supercontext:
we consider the functor $h_{\bP^{m|n}}:\salg \lra \sets$, where
$h_{\bP^{m|n}}(\cA)$ is defined as the set
the projective $\cA$-modules of rank one in $\cA^{m|n}:=
\cA \otimes \C^{m|n}$.
This is  the functor of points of
the {\sl projective superspace} described in Example \ref{projsup}.

\end{example}

The next question that  we want to tackle is how we can define an
embedding of a (super)variety into the projective (super)space
using the functor of points notation.

Let $X$ be a projective supervariety and
$\Phi:X \lra \bP^{m|n}$ be an injective morphism. As we discussed in
Example \ref{projsup} this embedding is encoded by a surjective morphism:
$$
\C[x_1, \dots, x_m,\xi_1 \dots, \xi_n]
\lra \C[x_1, \dots x_m,\xi_1 \dots,\xi_n]/(f_1, \dots, f_r)
$$
In the notation of
the functor of points, $\Phi$ is a {\it natural transformation}
between the two functors $h_X$ and $h_{\bP^{m|n}}$, given by
$$
\Phi_{\cA}: h_X(\cA) \lra h_{\bP^{m|n}}(\cA)
$$
with $\Phi_{\cA}$ injective.

If $\cA$ is a local
superalgebra, then an $\cA$-point $(a_1 \dots, a_m,\al_1 \dots ,\al_n)
\in h_{\bP^{m|n}(\cA)}$ is in $\phi_\cA(h_X(\cA))$
if and only if it satisfies the homogeneous polynomial relations
$$
\begin{array}{c}
f_1(a_1 \dots a_m,\al_1 \dots, \al_n)=0, \\
\vdots \\
f_r(a_1 \dots a_m,\al_1 \dots ,\al_n)=0.
\end{array}
$$
(See \cite{cfl} for more details).

In summary, to determine the coordinate superalgebra of a
projective supervariety with
respect to a certain projective embedding, we need
to check the relations satisfied by the coordinates {\sl just on local
superalgebras}. This will be our starting point
when we shall determine the coordinate superalgebra of the Grassmannian
supervariety with respect to its Pl\"{u}cker embedding.

\section{The chiral conformal superspace}
\label{conformal}

We are interested in  the super Grassmannian of $(2|0)$-planes inside
the superspace $\C^{4|1}$, that we denote with $\Gr$. This will be
our chiral conformal superspace once we establish an action of the
conformal supergroup on it.

\medskip

$\Gr$ is defined via its functor of points. For a generic superalgebra $\cA$,
the $\cA$-points of $\Gr$ consist of the projective modules of rank
$2|0$ in $\cA^{4|1}:=\cA \otimes \C^{4|1}$. It is not immediately clear that
this is the functor of points of a supervariety, however a fully
detailed proof of this fact is available in \cite{cfl}, Appendix A.
Another important issue is the fact that once a supervariety is given,
its functor of points is completely determined just by looking at
the {\sl local} superalgebras, and similarly the natural
transformations are determined if we know them for local superalgebras.
This a well known fact that
can be found for example in Ref. \cite{fg}, Appendix A.

\medskip

On a local superalgebra $\cA $,  $h_{\Gr}(\cA )$ consists of free submodules
of rank $2|0$ in $\cA ^{4|1}$ (on local superalgebras, projective
modules are free).
One such module can be specified by a
couple of independent even vectors, $a$ and $b$, which in the canonical basis
$\{e_1,e_2,e_3,e_4,\ep_5\}$ are given by two column vectors that span the subspace
\be\pi=\langle a,b\rangle=
\left\langle\begin{pmatrix} a_1 \\ a_2\\ a_3\\ a_4 \\
\alpha_5 \end{pmatrix}, \begin{pmatrix} b_1 \\ b_2\\ b_3\\ b_4 \\
\beta_5 \end{pmatrix}\right\rangle,\label{fopGr}\ee
with $a_i, b_i\in \cA _0$ and $\alpha_5,\beta_5\in \cA _1$. Let
\be h_{\rGL(4|1)}(\cA )=
\left\{\begin{pmatrix}c_{11}&c_{12}&c_{13}&c_{14}&\rho_{15}\\
c_{21}&c_{22}&c_{23}&c_{24}&\rho_{25}\\
c_{31}&c_{32}&c_{33}&c_{34}&\rho_{35}\\
c_{41}&c_{42}&c_{43}&c_{44}&\rho_{45}\\
\delta_{51}&\delta_{52}&\delta_{53}&\delta_{54}&d_{55}
\end{pmatrix}\right\},\label{fopSL}\ee
define the functor of points of the supergroup $\rGL(4|1)$,
where $c_{ij}, d_{55}\in \cA _0$ and $\rho_{i5},\delta_{5i}\in \cA _1$.
We can describe the action of the supergroup  $\rGL(4|1)$ over $\Gr$
as a natural transformation of the functors (for $\cA$ local),
$$
\begin{array}{ccc}
h_{\rGL(4|1)}(\cA )\times h_{\Gr}(\cA ) & \longrightarrow & h_{\Gr}(\cA )  \\
g, \langle a, \, b \rangle & \longmapsto & \langle g \cdot a, \, g \cdot
b \rangle .
\end{array}
$$

\medskip

Let $\pi_0=\langle e_1,e_2\rangle \in h_\Gr(A)$. The stabilizer
of this point in  $\rGL(4|1)$
is the upper parabolic super subgroup $P_u$, whose
functor of points is
$$
h_{P_u}(\cA )=
\left\{\begin{pmatrix}c_{11}&c_{12}&c_{13}&c_{14}&\rho_{15}\\
c_{21}&c_{22}&c_{23}&c_{24}&\rho_{25}\\
0&0&c_{33}&c_{34}&\rho_{35}\\0&0&c_{43}&c_{44}&\rho_{45}\\
0&0&\delta_{53}&\delta_{54}&d_{55}\end{pmatrix}\right\}
\subset h_{\rGL(4|1)}(A).
$$
Then, the Grassmannian is identified with the quotient
$$h_{\Gr}(\cA )=h_{\rGL(4|1)}(\cA )/h_{P_u}(\cA ).
$$

\medskip

We want now to work out the expression for the
\textit{Pl\"{u}cker embedding},
It is important to stress that, contrary
to what happens in the classical setting, in the super context
we have that a generic Grassmannian supervariety does not
admit a projective embedding. However for this particular
Grassmannian such embedding exists, as we are going to show presently.

\medskip

We want to give a natural transformation among the functors
$$
p:h_{\Gr}
\rightarrow h_{\bP(E)},
$$
where $E$ is the super vector space $E=\wedge^2\C^{4|1}\approx \C^{7|4}$.
Given the canonical basis for $\C^{4|1}$ we construct a basis for $E$
\begin{align}&e_1\wedge e_2, e_1\wedge e_3, e_1\wedge e_4, e_2\wedge e_3,
e_2\wedge e_4, e_3\wedge e_4, \ep_5\wedge \ep_5,\qquad &\hbox{(even)}\nonumber
\\ &e_1\wedge\ep_5, e_2\wedge\ep_5,e_3\wedge \ep_5,e_4\wedge\ep_5,
\qquad &\hbox{(odd)}\label{canonical basis}
\end{align}
As in the super vector space case, if $L$ is a $\cA$-module,
for $\cA \in \salg$,
we can construct $\wedge^2L $
$$
\Lambda^2 L= L\otimes L/
\langle u \otimes v +(-1)^{|u||v|}v \otimes u\rangle, \qquad u,v\in L.
$$
If $L \in h_{\Gr}(\cA )$, then  $\wedge^2L \subset \wedge^2\cA ^{4|1}$.
It is clear that if $L$ is a projective $\cA $-module of rank $2|0$, then
$\wedge^2L$ is a projective
$\cA $-module of rank
$1|0$. In other words it is an element of $h_{\bP(E)}(\cA )$, for
$E=\wedge^2 \C^{4|1}$. Hence we have defined a natural
transformation:
$$
\begin{CD}
h_{\Gr}(\cA ) @>p>> h_{\bP(E)}(\cA ) \\
L @>>> \wedge^2L.
\end{CD}
$$
Once we have the natural transformation defined, we can again restrict
ourselves to work only on local algebras.

\medskip

Let $a,b$ be two even independent vectors in $\cA^{4|1}$. For any
superalgebra $\cA$, they generate a free submodule of $\cA^{4|1}$ of
rank $2|0$. The natural transformation  described above
is as follows.
$$
\begin{CD}
 h_{\Gr}(A) @>p_{\cA}>>  h_{\bP(E)}(\cA) \\
 \langle a,b\rangle_{\cA} @>>>  \langle a\wedge b\rangle.
\end{CD}
$$
The map $p_{\cA}$ is clearly injective. The image $p_{\cA}(h_{\Gr}(\cA))$
is the subset of even elements in $h_{\bP(E)}(\cA)$ decomposable in
terms of two even vectors of $\cA^{4|1}$. We are going to find the necessary
and sufficient conditions for an  even  element $Q\in
h_{\bP(E)}(\cA)$
to be decomposable. Let
\begin{eqnarray}
&& Q=q+\lambda \wedge \ep_5+a_{55}
\ep_5 \wedge \ep_5 ,\quad
\hbox{with}\nonumber\\
&&q=q_{12}e_1\wedge e_2+\cdots +q_{34}e_3\wedge e_4,\quad
q_{ij}\in \cA_0, \nonumber\\
 &&\lambda=\lambda_1e_1+\cdots +\lambda_4 e_4,\quad \lambda_i\in
 \cA_1.\label{coordinates}
\end{eqnarray}
$Q$ is decomposable if and only if
\begin{eqnarray*}
&&
Q=(r+\xi\ep_5) \wedge (s+\theta\ep_5)\quad \hbox{with}\\
&& r=r_1e_1+\cdots r_4e_4,\quad s=s_1e_1+\cdots s_4e_4,\quad r_i,
s_i\in \cA_0\quad\xi, \theta\in \cA_1,
\end{eqnarray*} which means
$$
Q=r \wedge s+(\theta r-\xi
 s)\wedge \ep_5+ \xi \theta\ep_5 \wedge \ep_5\;
 \hbox{equivalent to}\;
 q=r \wedge s, \quad \lambda=\theta r-\xi s, \quad
 a_{55}=\xi\theta.
$$
These  are equivalent to the following:
$$
q \wedge q=0, \qquad q \wedge \lambda=0, \qquad
\lambda\wedge\lambda=2a_{55}q \qquad \lambda a_{55}=0.
$$
Plugging (\ref{coordinates})  we obtain
\begin{align}
&q_{12}q_{34}-q_{13}q_{24}+q_{14}q_{23}=0, &&
\hbox{ (classical Pl\"{u}cker relation)} \nonumber  \\&q_{ij}\lambda_k-q_{ik}\lambda_j+q_{jk}\lambda_i=0,&& 1\leq i<j<k\leq 4\nonumber \\
& \lambda_i \lambda_j=a_{55}q_{ij}&& 1\leq i<j\leq 4 \nonumber\\
& \lambda_ia_{55}=0.\label{superplucker} &&
\end{align}
These are the {\sl super Pl\"{u}cker relations}.
As we shall see in the next section the superalgebra
\be
\cO({\Gr})=k[q_{ij}, \lam_k, a_{55}]/\cI_P, \label{sgr}
\ee
is associated to the supervariety $\Gr$ in the Pl\"{u}cker embedding
described above, where
$\cI_P$ denotes the ideal of the super Pl\"{u}cker relations
(\ref{superplucker}). In other words $\cI_P$ contains all the
relations involving the coordinates $q_{ij}$, $\lambda_k$ and $a_{55}$.

\begin{remark}
The superalgebra $\cO(\Gr)$ is a sub superalgebra (though not a Hopf sub
superalgebra) of $\cO(\rGL(4|1))$. It is in fact the superalgebra
generated by the corresponding minors, and the Pl\"{u}cker
relations are all the relations satisfied by these minors
in $\cO(\rGL(4|1))$.
\end{remark}

\section{The super Grassmannian via invariant theory}
\label{invar}

In this section we propose an alternative and equivalent
way to construct the super Grassmannian $\Gr$ as a complex
supervariety and we give the coordinate superring associated to the
super Grassmannian in the Pl\"{u}cker embedding, thus completing
the discussion initiated in the previous section.

\medskip
As we have seen in Section \ref{supergeo}, the super Grassmannian
can be equivalently understood as a
a pair consisting of the underlying topological
space $G(2,4)$,  and a sheaf of superalgebras
conveniently chosen that we shall describe presently.

We recall first what happens in the ordinary case.
Let the set $S$ be
$$S=\{(v,w)\in \C^4\oplus \C^4\;/\; \rank(v,w)=2\},$$ and consider  the
equivalence relation
$$
(v,w)\sim(v',w')\quad \Leftrightarrow \quad
\rspan\{v,w\}=\rspan\{v',w'\},
$$
or equivalently
$$
(v,w)\sim(v',w')\quad \Leftrightarrow \quad \exists \, g\in
\rGL(2,\C)\; \hbox{such that} \;(v',w')=(v,w)g.
$$
Then we have that $G(2,4)=S/\sim$.

We consider now the set of  polynomials on $S$, $\rPol(S)$, and the subset of such polynomials that  is semi-invariant under the transformation of $\rGL(2,\C)$, that is
$$f(v', w')=f(u,v)\lambda(g),\qquad \lambda(g) \in \C, \quad f\in \rPol(S).$$  This defines the homogeneous ring of  $G(2,4)$, which  is generated by the  six determinants \cite{fu}.
$$y_{ij}=v_iw_j-v_jw_i,\quad \hbox{with } i<j\; \hbox{ and }\; \lambda =\det g.$$ These are not all independent, they satisfy the Pl\"{u}cker relation
$$y_{12} y_{34}+y_{23} y_{14}+y_{31}
y_{24}=0.$$

Let $\cH$ be the sheaf of polynomials on $S$,
so for each open set in $\tilde U\subset S$,
$\cH(\tilde U)=\rPol(\tilde U)$ and
$\cH^{\mathrm{inv}}$  the  subsheaf of $\cH$ corresponding
to the semi-invariant polynomials.

Let $\pi:S\rightarrow G(2,4)$ be the natural projection. It is clear that for $U\subset_{\mathrm{open}}G(2,4)$, then $\tilde U=\pi^{-1}( U)\subset S$ is also open
 in $S$. We can define the following sheaf over $G(2,4)$:
$$\cO(U)=\cH^{\mathrm{inv}}(\pi^{-1}( U)).$$
This is the structural sheaf of the projective variety $G(2,4)$
with respect to the Pl\"{u}cker embedding.

 \medskip

Now we turn to the super setting and we want to define
the sheaf of superalgebras generalizing the non super
construction to the super Grassmannian.  We define the  superalgebra
$$
\cF(S):=\rPol(S) \otimes\Lambda[\xi_1,\xi_2].$$
Let $(v,w)\in S$ and  consider the $(5 \times 2)$ matrix
$$
\begin{pmatrix}v&w\\\xi_1&\xi_2\end{pmatrix}=
\begin{pmatrix}v_1&w_1\\\vdots&\vdots\\v_4&w_4\\\xi_1&\xi_2\end{pmatrix}.
$$
The group  $\rGL(2,\C)$  acts on
the right on these matrices
$$
\begin{pmatrix}v'&w'\\\xi'_1&\xi'_2\end{pmatrix}=
\begin{pmatrix}v&w\\\xi_1&\xi_2\end{pmatrix}\cdot
g,\qquad g\in\rGL(2,\C).$$
We will write an element $f(v,w,\xi)\in \cF(S)$ as
$$
f(v,w,\xi)=\sum_{i,j=0,1}f_{ij}(v,w)\xi_1^i\xi_2^j.
$$
We will refer to the elements of $\cF(S)$ as `functions', being this customary in the physics literature. We now consider the set of semi-invariant functions
$$f(v',w',\xi')=f(v,w,\xi)\lambda(g), \qquad \lambda(g) \in \C, \quad f\in \cF(S).$$
The following
functions are semi-invariant:
\begin{equation}y_{ij}=v_iw_j-v_jw_i, \quad
\theta_i=v_i\xi_2-w_i\xi_1, \quad
a=\xi_1\xi_2,\label{superinvariants}\end{equation}  with $\lambda(g)=\det g$ but they are
not all independent. They satisfy the {\it super Pl\"{u}cker relations}
(\ref{superplucker})
\begin{align*} &y_{12}y_{34}-y_{13}y_{24}+y_{14}y_{23}=0, \qquad
&&\hbox{ (standard Pl\"{u}cker relation)} \\
&y_{ij}\theta_k-y_{ik}\theta_j+y_{jk}\theta_i=0&&1\leq i<j<k\leq 4\\& \theta_i\theta_j=ay_{ij}&&1\leq i<j\leq 4\\&\theta_ia=0&&1\leq i\leq 4=0.
\end{align*}

We want to show that the elements in (\ref{superinvariants}) generate
the ring of semi-invariants and that (\ref{superplucker}) are all
the relations among these generators.
\begin{proposition}
Let $f$ be a homogeneous semi-invariant function, so
$$f(v',w', \xi')=f(v,w, \xi)\lambda(g)$$  with
$$\begin{pmatrix}v'&w'\\\xi'_1&\xi'_2\end{pmatrix}=
\begin{pmatrix}v&w\\\xi_1&\xi_2\end{pmatrix}\cdot
g,\qquad g\in\rGL(2,\C).$$
Then in the decomposition
\be f(v,w, \xi)=f_0(v,w)+\sum_if_i(v,w)\xi_i+f_{12}(v,w)\xi_1\xi_2,\label{decomposition}\ee
one has that  $f_0(v,w)$ and $f_{12}(v,w)$ are standard (non-super) semi-invariants  and
$$\sum_if_i(v,w)\xi_i=\sum_ih_i(v,w)\theta_i,$$
with $h_i(v,w)$ also a standard semi-invariant.
\end{proposition}

\begin{proof} Let us take
$$g=\begin{pmatrix}a&b\\c&d\end{pmatrix}, \quad \hbox{so}\quad \begin{pmatrix}v'&w'\\\xi'_1&\xi'_2\end{pmatrix}=\begin{pmatrix}va+wc&vb+wd\\\xi_1a+\xi_2c&\xi_1b+\xi_2d\end{pmatrix}.$$
Then we can see immediately that each term in
(\ref{decomposition}) has to be a semi-invariant, so
\begin{align*}&f_0(v',w')=\lambda(g)f_0(v,w),\qquad \sum_if_i(v',w')\xi'_i=\lambda(g)\sum_if_i(v,w)\xi_i,\\ &f_{12}(v',w')\xi'_1\xi'_2=f_{12}(v,w)\xi_1\xi_2.\end{align*}
We have that $f_0$ is an ordinary semi-invariant transforming with $\lambda(g)$, and since $\xi_1'\xi_2'=\xi_1\xi_2\det g$,  $f_{12}(v,w)$ is a
ordinary semi-invariant transforming with $\lambda(g)\det g^{-1}$.
   The odd terms $\theta^i$
are of the same form as the ordinary invariants $y_{ij}$, since
the fact that $\xi_i$ is odd plays no particular role here
(recall that we are considering the action of an ordinary
group, namely $\rGL(2,\C)$).
So by the same argument we have in the ordinary case,
there are no other odd invariants, besides those
we have already found, that are linear in the odd variable $\xi_1$ and $\xi_2$.
Then
$$
\sum_if_i\xi_i=\sum_ih(v,w)_i\theta_i,
$$ where $h(v,w)_i$ transforms with
$\lambda(g)\det g^{-1}$.
\end{proof}

We now wish to give a result that describes completely
the relations among the invariants.

Consider the polynomial superalgebra $\C[a_{ib}]$,
$1 \leq i \leq 5$, $1 \leq b \leq 2$, with their parity defined as
$$p(a_{ij})=p(i)+p(j),\quad \hbox{with}\quad p(k)=0\hbox{ if }
0 \leq k \leq 4\hbox{ and } p(5)=1.$$

On $\C[a_{ij}]$ there exists the following action of $\rGL(2,\C)$ :
$$
\begin{CD}
 \C[a_{ib}]\times \rGL(2,\C) @>>> \C[a_{ib}]\\
( a_{ia}, g^{-1}) @>>> \sum_k a_{ib}g^{-1}_{ba}
\end{CD}
$$
We have just proven that the semi-invariants are generated by the polynomials
\begin{align*}
&d_{ij}=a_{i1}a_{j2}-a_{i2}a_{j1}, \quad  1\leq i<j \leq 5, \\
&d_{55}=a_{51}a_{52}.
\end{align*}

We have the following proposition:

\begin{proposition} \label{presentationR}
Let $\cO(\Gr)$ be the subring  of $\C[a_{ib}]$ generated by the
determinants $d_{ij}=a_{i1}a_{j2}-a_{j1}a_{i2}$ and
$d_{55}=a_{51}a_{52}$.
Then $\cO(\Gr) \cong \C[a_{ib}]/I_P$, where $I_P$ is the ideal of the super Pl\"{u}cker relations (\ref{superplucker}). In other words $I_P$ contains all
the possible relations satisfied by $d_{ij}$ and $d_{55}$.
\end{proposition}
\begin{proof}
It is easy to verify that  $d_{ij}$ and $d_{55}$ satisfy
all the above relations, the problem is to prove that these are the only
relations.

The proof of this fact is the same as in the classical setting.
Let us briefly sketch it.
Let $I_1, \dots ,I_r$ be multiindices organized in a tableau.
We say that a  tableau is {\it superstandard} if it is strictly
increasing along rows with the exception of the number 5
(that can be repeated)
and weakly increasing along columns.
A {\it standard monomial} in $\cO(\Gr)$ is
a monomial $d_{I_1}, \cdots, d_{I_r}$ where the indices $I_1, \dots,
I_r$ form a superstandard tableau. Using the
super Pl\"{u}cker relation one can verify that any monomial in $\cO(\Gr)$ can
be written as a linear combination of standard ones. This can be
done directly or using the same argument for the classical case
(see Ref. \cite{fu} pg 110 for more details). The standard monomials
are also linearly independent, hence they form a basis for $\cO(\Gr)$ as
$\C$-vector space. Again this is done with the same argument as in Ref.
\cite{fu} pg 110. So given a relation in $\cO(\Gr)$, once we write each
term as a standard monomial we obtain that either the relation is
identically zero (hence it is a relation in the Pl\"{u}cker ideal) or
it gives a relation among the standard monomials, which gives a
contradiction.
\end{proof}

In the end we summarize the main results of Sections \ref{conformal} and
\ref{invar} with a corollary.

\begin{corollary} \label{superpres}
\begin{enumerate}
\item Let $\Gr$ be the Grassmannian of $2|0$ spaces in $\C^{4|1}$.
Then $\Gr \subset \bP^{7|4}$, that is $\Gr$
is a projective supervariety. Such embedding is encoded by the superring
$\cO(\Gr)$ described above.
\item $\cO(\Gr)$ is isomorphic to the ring generated by the determinants
$d_{ij}$, $d_{55}$.
\end{enumerate}
\end{corollary}

\section{The chiral Minkowski superspace}
\label{smink}

In this section we concentrate our attention to determine
the big cell inside the Grassmannian supervariety that we have
discussed in the previous sections. We shall identify such big cell
with the chiral Minkowski superspace.

\medskip

As in the ordinary setting, the super Grassmannian $\Gr$
admits an open cover in terms of affine superspaces: topologically
the two covers are the same.

\medskip

We want to describe
the functor of points of the {\sl big cell} $U_{12}$ inside $\Gr$.
This is the open affine functor corresponding to
the points in which the coordinate $q_{12}$ is invertible.

\medskip

First of all,
we write an element of $h_{\rGL(4|1)}(\cA)$ in blocks as (see (\ref{fopSL}))
$$
\begin{pmatrix}
C_1&C_2&\rho_1\\C_3&C_4&\rho_2\\\delta_1&\delta_2&d_{55}
\end{pmatrix}.
$$ Assuming that $\det C_1$ is invertible,
we can bring this matrix, with a transformation of $h_{P_u}(\cA)$,
to the form
\be
\begin{pmatrix}
C_1&C_2&\rho_1\\C_3&C_4&\rho_2\\\delta_1&\delta_2&d_{55}
\end{pmatrix}h_{P_u}(\cA)=
\begin{pmatrix}\id_2&0&0\\A&\id_2&0\\\alpha&0&1\end{pmatrix}
h_{P_u}(\cA) \, \in \, h_{\rGL(4|1)}(\cA)\, \big/ \, h_{P_u}(\cA)
\label{standardform}\ee

\medskip

Consider the subspace $\pi=\rspan\{ a, \, b \}$ in $h_\Gr(\cA)$ for
$\cA$ local. Recall that in Sec. \ref{conformal} we
made the identification:
$h_\Gr(\cA) \cong   h_{\rGL(4|1)}(\cA)\, \big/ \, h_{P_u}(\cA)$. Hence:
$$
\pi = \rspan \{ a, \, b \} \approx
\begin{pmatrix}
C_1&C_2&\rho_1\\C_3&C_4&\rho_2\\\delta_1&\delta_2&d_{55}
\end{pmatrix}h_{P_u}(\cA) \, \in \, h_{\rGL(4|1)}(\cA)\, \big/ \, h_{P_u}(\cA)
$$
with $\det C_1$ invertible. Then, by a change of coordinate (\ref{standardform})
we can bring this matrix to the standard form detailed above
$$
\pi \approx
\begin{pmatrix}\id_2&0&0\\A&\id_2&0\\\alpha&0&1\end{pmatrix}
h_{P_u}(\cA),
\qquad A=
\begin{pmatrix}a_{11}&a_{12}\\a_{21}&a_{22}\end{pmatrix},
\qquad \alpha =(\alpha_1,\alpha_2),
$$ with the entries of $A$ in $\cA_0$ and the entries of $\alpha$ in $\cA_1$. Its column vectors generate also the submodule $\langle a, \, b \rangle$.

\medskip

The assumption that $\det C_1$ is invertible is equivalent to assume to be in
the topological open set $|U_{12}|=|\Gr| \cap |V_{12}|$, where $V_{12}$ is the
affine open set corresponding to the topological open set $|V_{12}|$
defined by taking in $\bP(E)$
the coordinate $q_{12}$ to be invertible. Consequently the coordinate
superring of the affine open subvariety
$U_{12}$ of $\Gr$ corresponds to the projective localization of
the Grassmannian superring in the coordinate $q_{12}$. In other
words it consists of the elements of degree zero in
$$
\C[q_{ij}q_{12}^{-1}, \lambda_jq_{12}^{-1}, a_{55}q_{12}^{-1}]
\subset \cO(\Gr)[q_{12}^{-1}].
$$
As one can readily check, there are no relations among these
generators so that
the big cell $U_{12}$ of $\Gr$ is the affine superspace
with coordinate ring 
\be\cO(U_{12})=\C[x_{ij},\xi_j]\approx\C^{4|2}.
\label{bigcellsuperalgebra}\ee
where we set $x_{ij}=q_{ij}q_{12}^{-1}$, $x_{55}=a_{55}q_{12}^{-1}$,
$\xi_j=\lambda_jq_{12}^{-1}$.

\medskip

We are now interested in the super subgroup of
${\rGL(4|1)}$ that preserves the big cell ${U_{12}}$.
This the lower parabolic sub-supergroup $P_l$ (see \cite{cfl}), whose functor
of points is given in suitable coordinates as
type
$$
h_{P_l}(\cA)=\left\{\begin{pmatrix}x&0&0\\tx&y&y\eta\\d \tau &d   \xi& d
\end{pmatrix}\right\} \subset h_{\rGL(4|1)}(\cA)
$$
where $x$ and $y$ are even, invertible $2\times 2$ matrices, $t$ is an even,
arbitrary $2\times 2$ matrix, $\eta$ a $2\times 1$ odd matrix, $\tau, \xi$
are $1\times 2$ odd matrices and  $d$ is an invertible even element.

The action of the supergroup $P_l$ on the big cell $U_{12}$
is as follows,
$$\begin{CD}h_{P_l}(\cA)\times h_{U_{12}}(\cA)@>>>\quad h_{U_{12}}(\cA)\\\\
\left(\begin{pmatrix}x&0&0\\tx&y&y\eta\\d \tau &d \xi& d \end{pmatrix},\begin{pmatrix}\id_2\\A\\\alpha\end{pmatrix}\right)@>>>\begin{pmatrix}\id_2\\A'\\\alpha'\end{pmatrix},
\end{CD}$$
where, using a transformation of $h_{P_u}(\cA)$ to revert the resulting matrix to the standard form (\ref{standardform}), we have
\be\begin{pmatrix}\id_2\\A'\\\alpha'\end{pmatrix}=\begin{pmatrix}\id_2\\y(A+\eta\alpha) x^{-1}+t\\
d (\alpha  +\tau +\xi A) x^{-1}\end{pmatrix}.\label{actionbigcell}\ee
The subgroup with $\xi=0$ is the super Poincar\'{e} group times dilations (compare with Eq. (14) in Ref \cite{flv}). In that case
$$d=\det x\det y.$$

\section{Quantum chiral conformal superspace}
\label{qconformal}

In this section we give a quantum deformation of $\cO(\Gr)$, discussed
in the previous sections. This will yield a quantum deformation of
the chiral conformal superspace together with the natural coaction
of the conformal supergroup on it.

\begin{definition}\label{maninrelations}
Let us define following Manin \cite{ma2} the quantum matrix superalgebra.
$$
M_q(m|n)=_{def}\C_q<a_{ij}>/I_M
$$
where $\C_q<a_{ij}>$ denotes the free algebra over $\C_q=\C[q,q^{-1}]$
generated by the homogeneous variables $a_{ij}$
and the ideal $I_M$ is generated by the relations \cite{ma2}:
$$
\begin{array}{c}
a_{ij}a_{il}=(-1)^{\pi(a_{ij})\pi(a_{il})}
q^{(-1)^{p(i)+1}}a_{il}a_{ij}, \quad j < l \\ \\
a_{ij}a_{kj}=(-1)^{\pi(a_{ij})\pi(a_{kj})}
q^{(-1)^{p(j)+1}}a_{kj}a_{ij}, \quad i < k \\ \\
a_{ij}a_{kl}=(-1)^{\pi(a_{ij})\pi(a_{kl})}a_{kl}a_{ij}, \quad
i< k,j > l \quad or \quad i > k,j < l \\ \\
a_{ij}a_{kl}-(-1)^{\pi(a_{ij})\pi(a_{kl})}a_{kl}a_{ij}= \\ \\
(-1)^{\pi(a_{ij})\pi(a_{kj})}(q^{-1}-q)a_{kj}a_{il} \quad i<k,j<l
\end{array}
$$
where $p(i)=0$ if $1 \leq i \leq m$, $p(i)=1$ otherwise
and $\pi(a_{ij})=p(i)+p(j)$ denotes the parity of $a_{ij}$.
\end{definition}

$M_q(m|n)$ is a bialgebra with the usual comultiplication and
counit:
$$
\Delta(a_{ij})=\sum a_{ik} \otimes a_{kj},
\qquad
\ep(a_{ij})=\de_{ij}.
$$
\medskip

We are ready to define the general linear supergroup which
will be most interesting for us.

\begin{definition}
We define {\it quantum general linear supergroup}
$$
\rGL_q(m|n)=_{def}
M_q(m|n)
\langle{D_{1}}^{-1},{D_{2}}^{-1}\rangle
$$
where ${D_{1}}^{-1}$,
${D_{2}}^{-1}$ are even indeterminates such
that:
$$
\begin{array}{c}
{D_1}D_1^{-1}=1=
{D_1}^{-1}D_{1},\qquad
{D_{2}}
{D_{2}}^{-1}=1=
{D_{2}}^{-1}
{D_{2}}
\end{array}
$$
and
$$
\begin{array}{c}
D_{1}=_{def}\sum_{\s \in S_m}(-q)^{-l(\s)}
a_{1\s(1)} \dots a_{m\s(m)} \\ \\
D_{2}=_{def}\sum_{\s \in S_n}(-q)^{l(\s)}
a_{m+1,m+\s(1)} \dots a_{m+n,m+\s(n)}
\end{array}
$$
are the quantum determinants of the diagonal blocks.
\end{definition}
$\rGL_q(m|n)$ is H hopf algebra, where the comultiplication and
counit are the same as in $M_q(m|n)$,
while the antipode $S$ is detailed in Ref.  \cite{fi5}.

\medskip

We now give the central definition in analogy with
the ordinary setting (compare with Prop. \ref{superpres}).

\begin{definition}
Let the notation be as above. We define {\it quantum super Grassmannian}
of $2|0$ planes in $4|1$ dimensional superspace as the non commutative
superalgebra $Gr_q$ generated by the following quantum super minors in
$\rGL_q(4|1)$:
$$
\begin{array}{c}
D_{ij}= a_{i1}a_{j2}-q^{-1}a_{i2}a_{j1}, \qquad 1\leq i<j \leq 4, \qquad
D_{55}=a_{51}a_{52} \\ \\
D_{i5}= a_{i1}a_{52}-q^{-1}a_{i2}a_{51},
\qquad 1 \leq i \leq 4. \\ \\
\end{array}
$$
\end{definition}

For clarity let us write all the generators:
$$
D_{12}, \quad D_{13}, \quad D_{14}, \quad D_{23}, \quad
D_{24}, \quad D_{34}, \quad D_{55}, \quad D_{15}, \quad D_{25},
\quad D_{35}, \quad D_{45}
$$

Notice that when $q=1$ this is the coordinate ring of the
super Grassmannian.

\medskip

We need to work out the commutation relations and the quantum
Pl\"{u}cker relations in order to be able to give a presentation
of the quantum Grassmannian in terms of generators and
relations.

\medskip

Let us start with the commutation relations. With
 very similar calculations to the ones in Ref. \cite{fi2} one
finds the following relations:

\begin{itemize}

\item
If $i,j,k,l$ are {\sl not}
all distinct we have ($1 \leq i,j,k,l \leq 5$):
$$
D_{ij}D_{kl}=q^{-1}D_{kl}D_{ij}, \qquad (i,j)<(k,l)
$$
where $<$ refers to the lexicographic ordering.

\item
If $i,j,k,l$ are instead all distinct we have:
$$
\begin{array}{c}
D_{ij}D_{kl}=q^{-2}D_{kl}D_{ij}, \qquad 1 \leq i<j<k<l \leq 5 \\ \\
D_{ij}D_{kl}=q^{-2}D_{kl}D_{ij}-(q^{-1}-q)D_{ik}D_{jl},
\qquad 1 \leq i<k<j<l \leq 5 \\ \\
D_{ij}D_{kl}=D_{kl}D_{ij}, \qquad 1 \leq i<k<l<j\leq 5
\end{array}
$$

\item

The only commutation relations that we are left to be shown
are the following:
$$
D_{ij}D_{55}, \qquad D_{i5}D_{j5}, \qquad D_{i5}D_{55}
$$

After some computations one gets:
$$
\begin{array}{c}
D_{ij}D_{55}=q^{-2}D_{55}D_{ij}, \qquad 1 \leq i <j \leq 4\\ \\
D_{i5}D_{j5}=-q^{-1}D_{j5}D_{i5}-(q^{-1}-q)D_{ij}D_{55}
\qquad 1\leq i<j \leq 4\\ \\
D_{i5}D_{55}=D_{55}D_{i5}=0,  \qquad 1 \leq i \leq 4.
\end{array}
$$

\end{itemize}

This concludes the discussion of the commutation relations.
As for the Pl\"{u}cker relations, using the result for the
non super setting (refer
to \cite{fi2}) we have
$$
\begin{array}{c}
D_{12}D_{34}-q^{-1}D_{13}D_{24}+q^{-2}D_{14}D_{23}=0 \\ \\
D_{ij}D_{k5}-q^{-1}D_{ik}D_{j5}+q^{-2}D_{i5}D_{jk}=0,
\qquad 1 \leq i<j<k \leq 4 \\
\end{array}
$$
To this we must add the relations, which can be computed
directly:
$$
D_{i5}D_{j5}=qD_{ij}D_{55}, \qquad 1 \leq i<j \leq 4.
$$

The next proposition summarizes all of our calculations
and the proof can be found in Ref.
\cite{cfl}.
\begin{proposition}${}^{}$

\begin{itemize}
\item The quantum Grassmannian ring is given in terms of
generators and relations as:
$$
Gr_q=\C_q\langle X_{ij}\rangle/I_{Gr}
$$
where $I_{Gr}$ is the two-sided ideal generated by the
commutations and Pl\"ucker relations in the indeterminates
$X_{ij}$. Moreover $Gr_q/(q-1) \cong \cO(\Gr)$
(see Section \ref{conformal}).

\item The quantum Grassmannian ring is the free ring
over $\C_q$ generated by the monomials in the quantum determinants:
$$
D_{i_1j_1}, \dots, D_{i_rj_r}
$$
where $(i_1,j_1), \dots, (i_r,j_r)$ form a semistandard tableau
(for its definition refer to \cite{cfl}).
\end{itemize}

\end{proposition}

The quantum Grassmannian that we have constructed
admits a coaction of the quantum supergroup
$\rGL_q(4|1)$.
The proof of the following proposition amounts to a direct
check (we refer again to Ref. \cite{cfl} for more details).

\begin{proposition}
$Gr_q$ is a quantum homogeneous superspace for the quantum supergroup
$\rGL_q(4|1)$, i. e., we have a coaction given via the restriction of the
comultiplication of $\rGL_q(4|1)$:
$$
\Delta|_{Gr_q}: Gr_q \lra  \rGL_q(4|1) \otimes Gr_q.
$$
\end{proposition}

\section{Quantum Minkowski superspace}
\label{qmink}

We now turn to the
quantum deformation of the big cell inside $Gr_q$;
it will be our model for the quantum Minkowski superspace.

\medskip

In Section \ref{smink} we wrote the action of the lower
parabolic supergroup $P_l$ using the functor of
points (\ref{actionbigcell}). We want  now to translate it into the
coaction language in order
to make the generalization to the quantum setting.

\medskip

Let $\cO(P_l)$ be the superalgebra:
$$
\cO(P_l):=\cO(\rGL(4|1))/\cI
$$
where $\cI$ is the (two-sided) ideal generated by
$$
g_{1j},g_{2j}, \quad \hbox{for}\quad
j=3,4 \quad \hbox{and}\quad \gamma_{15},\gamma_{25}.
$$
This is the Hopf superalgebra coordinate superring
of the lower parabolic subgroup $P_l$, with
comultiplication naturally inherited by
$\cO(\rGL(4|1))$.

In matrix form, for $\cA$ local, we have
\be
h_{P_l}(\cA)=\left\{
\begin{pmatrix}
g_{11} & g_{12} & 0 & 0 & 0 \\
g_{21} & g_{22} & 0 & 0 & 0 \\
g_{31} & g_{32} & g_{33} & g_{34} & \ga_{35} \\
g_{41} & g_{42} & g_{43} & g_{44} & \ga_{45} \\
\ga_{51} & \ga_{52} & \ga_{53} & \ga_{54} & g_{55} \\
\end{pmatrix}
\right\} \subset h_{\rGL(m|n)}(\cA).\label{lowparabolic}
\ee
The superalgebra representing the big cell $U_{12}$
can be realized as a subalgebra of $\cO(P_l)$. In order to see this better, let
us make the following two different changes of variables in $P_l$:
\be
\begin{pmatrix}
g_{11} & g_{12} & 0 & 0 & 0 \\
g_{21} & g_{22} & 0 & 0 & 0 \\
g_{31} & g_{32} & g_{33} & g_{34} & \ga_{35} \\
g_{41} & g_{42} & g_{43} & g_{44} & \ga_{45} \\
\ga_{51} & \ga_{52} & \ga_{53} & \ga_{54} & g_{55} \\
\end{pmatrix}=\begin{pmatrix}
x      & 0 & 0 \\
tx     & y & y\eta\\
\ttau x & d\xi & d \\
\end{pmatrix}=
\begin{pmatrix}
x      & 0 & 0 \\
tx     & y & y\eta\\
d\tau  & d\xi & d \\
\end{pmatrix} \label{coordchange}
\ee
Notice that the only difference between
the two sets of variables is that we replace $\tau$ with $\tilde \tau$
and we have:
\be d\tau=\ttau x,\label{ttau}\ee

The next proposition tells us that these are sets of generators
for $\cO(P_l)$ and that
having $\ttau$ is essential to describe the big cell. Again for
the proof we refer the reader to Ref. \cite{cfl}, while the explicit
expressions for the generators come from a direct calculation.

\begin{proposition}${}^{}$\label{coaction on big cell}

1. The Hopf superalgebra $\cO(P_l)$ is generated by the following sets
of variables:
\begin{itemize}
\item $x$, $y$, $t$, $\ttau$, $\xi$, $\eta$ and $d$;
\item $x$, $y$, $t$, $\tau$, $\xi$, $\eta$ and $d$
\end{itemize}
defined as
\be
\begin{array}{cc}
x =\left( \begin{array}{cc}
g_{11} & g_{12}\\ g_{21} & g_{22} \end{array} \right),
 &
y =\left( \begin{array}{cc}
g_{33} & g_{34}\\ g_{43} & g_{44} \end{array} \right), \\ \\
t =\left(\begin{array}{cc} -d_{23}d_{12}^{-1} & d_{13}d_{12}^{-1} \\
-d_{24}d_{12}^{-1} & d_{14}d_{12}^{-1} \end{array} \right) &
d=g_{55} \\ \\
\ttau =( -d_{25}d_{12}^{-1} , d_{15}d_{12}^{-1}) &
\tau=( g_{55}^{-1}\ga_{51}, g_{55}^{-1}\ga_{52}) \\ \\
\eta=\begin{pmatrix} {d^{34}_{34}}^{-1}\ga_{35} \\
{d^{34}_{34}}^{-1}\ga_{45} \end{pmatrix}
& \xi=\begin{pmatrix}g_{55}^{-1}\ga_{53} &
g_{55}^{-1}\ga_{54}  \end{pmatrix}
\end{array} \label{explicitcoord}
\ee
where for $1 \leq i<j \leq 4$
$$
d_{ij}=g_{i1}g_{j2}-g_{j1}g_{i2}, \qquad
d_{i5}=g_{i1}\ga_{52}-\ga_{51}g_{i2}, \qquad
d^{34}_{34}=g_{33}g_{44}-g_{34}g_{43}.
$$

\medskip

2. The subalgebra of  $\cO(P_l)$ generated by $(t, \ttau)$
coincides with  the big cell superring
$\cO(U_{12})$ as defined in (\ref{bigcellsuperalgebra}). It is given
by the projective localization of $\cO(\Gr)$ with respect to $d_{12}$.

\medskip

3. There is a well defined
coaction $\tilde\Delta$ of
$\cO({P_l})$ on $\cO(U_{12})$ induced by the coproduct
in $\cO({P_l})$,
$$
\begin{CD}
\tilde\Delta:\cO(U_{12})@>\tilde \Delta>>\cO({P_l}) \otimes \cO(U_{12})
\end{CD}
$$
which explicitly takes the form:
\begin{align*}
\tilde\Delta t_{ij}&=t_{ij}\otimes 1+y_{ia} S(x)_{bj}\otimes t_{ab}+
y_{i}\eta_aS(x)_{bj}\otimes \ttau_{jb},&\nonumber \\ \\
\tilde\Delta\ttau_j=&(d\otimes 1)(\tau_a\otimes 1+
\xi_b\otimes t_{ba}+1\otimes \ttau_a)(S(x)_{aj}\otimes 1),&\nonumber
\end{align*}
The reader should notice right away that this is the dual to
the expression (\ref{actionbigcell}).
\end{proposition}
\medskip

We now turn to the quantum setting.
In order to keep our notation minimal, we use the same letters
as in the classical case  to denote the generators of the
quantum big cell and the quantum supergroups.

\medskip

Let $\cO(P_{l, q})$ be the superalgebra:
$$
\cO(P_{l, q}):=\cO(\rGL_q(4|1))/\cI_q
$$
where $\cI_q$ is the (two-sided) ideal in $\cO(\rGL_q(4|1))$
generated by
\be g_{1j},g_{2j}, \quad \hbox{for}\quad j=3,4
\quad \hbox{and}\quad \gamma_{15},\gamma_{25}.\label{ideal}\ee This is
the Hopf superalgebra of the lower parabolic subgroup, again with
comultiplication the one naturally inherited from
$\cO(\rGL_q(4|1))$.

As in the classical case, it is convenient to change coordinates
exactly in the same way (see \ref{coordchange}),
this time, however, paying extra attention to the
order in which we take the variables.
We can write the new coordinates for $\cO(P_{l,q})$ explicitly:

\begin{align*}
&x =\begin{pmatrix}
g_{11} & g_{12}\\ g_{21} & g_{22}  \end{pmatrix},
&&
t =
\begin{pmatrix} -q^{-1}D_{23}D_{12}^{-1} & D_{13}D_{12}^{-1}\\
-q^{-1}D_{24}D_{12}^{-1} & D_{14}D_{12}^{-1}  \end{pmatrix}
\nonumber \\\nonumber \\
&y = \begin{pmatrix}
g_{33} & g_{34}\\ g_{43} & g_{44} \end{pmatrix},&&
d =g_{55}, \\ \\
&
\ttau=\begin{pmatrix}  -q^{-1}D_{25}D_{12}^{-1} &
D_{15}D_{12}^{-1}  \end{pmatrix}, &
&\xi=\begin{pmatrix}g_{55}^{-1}\ga_{53} &
g_{55}^{-1}\ga_{54}  \end{pmatrix}
\\ \end{align*}
$$
\eta =y^{-1} \begin{pmatrix} \ga_{35} \\ \ga_{45}  \end{pmatrix}=
{(D_{34}^{34})}^{-1}
\begin{pmatrix} g_{44} & -q^{-1}g_{34}\\ -qg_{43} & g_{33}
 \end{pmatrix}=
\begin{pmatrix}
-q^{-1}{D^{34}_{34}}^{-1}D^{45}_{34} \\
{D^{34}_{34}}^{-1}D_{34}^{35} \\
\end{pmatrix}
$$

It is not hard to see that $\cO(P_{l,q})$ is also generated by
$x, y, d,\eta, \xi$ and $\ttau$.

\begin{remark} \label{qsuperpoinc}
The {\it quantum Poincar\'{e}  supergroup times dilations} is the
quotient of $\cO(P_{l,q})$  by the ideal $\xi=0$.
In fact as one can readily check with a simple calculation,
if $\cO(Po)$ denotes the function algebra of the super
(unquantized) Poincar\'{e}
groups times dilations, we have
that
$$
\left( \, \cO(P_{l,q}) \, / \, (\xi) \right) / (q-1) \cong \cO(Po).
$$
One can also easily check
that $(\xi)$ is a Hopf ideal, so the comultiplication goes to the quotient.
The quantum  Poincar\'{e} supergroup times dilations is then generated
by the images in the quotient of
$x, y, d,\eta$ and $\ttau$.  In matrix form, one has
$$\begin{pmatrix}
x      & 0 & 0 \\
tx     & y & y\eta \\
\ttau x & 0 & d  \\
\end{pmatrix}.
$$
Explicitly in these coordinates its presentation is given as follows:
$$
\cO(P_{l,q}) \, / \, (\xi) =\C_q<t,x,y,\eta,\tau>/I_{Po,q}
$$
where $I_{Po,q}$ is the ideal generated by the following relations. The
indeterminates
$x$ and $y$ behave respectively
as quantum (even) matrices, that is, their entries are
subject to the relations \ref{maninrelations}. In other words we have
for $x$ (and similarly for $y$):
$$
\begin{array}{c}
x_{11}x_{12}=q^{-1} x_{12}x_{11}, \quad
x_{11}x_{21}=q^{-1} x_{21}x_{11}, \quad
x_{21}x_{22}=q^{-1} x_{22}x_{21} \\ \\
x_{12}x_{22}=q^{-1} x_{22}x_{12}, \quad
x_{12}x_{21}=x_{21}x_{12}, \quad
x_{11}x_{22}-x_{22}x_{11}=(q^{-1}-q)x_{12}x_{21}
\end{array}
$$
Moreover the entries in $x$ and $y$ commute with each other.
$x$ and $t$, $\ttau$ commute in the following way. Let $i=1,2$,
$j=3,4$.
$$
\begin{array}{c}
x_{1i} t_{j1}=q^{-1} t_{j1} x_{1i}, \qquad
x_{2i} t_{j1}=  t_{j1} x_{2i}, \\ \\
x_{1i} \ttau_{51}=q^{-1} \ttau_{51} x_{1i}, \qquad
x_{2i} \ttau_{51}=  \ttau_{51} x_{2i},
\\ \\
x_{1i} \ttau_{52}=\ttau_{52} x_{1i}, \qquad
x_{2i} \ttau_{52}=q^{-1} \ttau_{52} x_{2i} \\ \\
\end{array}
$$
$x$ commutes with $\eta$ and $d$.
$y$, $t$ and $\ttau$ satisfy similar relations as $x$, $t$ and $\tau$
that we leave to the reader as an exercise (the rows are exchanged with
the columns). $y$ and $\eta$ commute following the rules of quantum
super matrices, very much the same calculation and relations expressed
in \ref{commt-tau}.
$y$ and $d$ commute.
The commutation among $t$ and $\ttau$ are expressed in prop. \ref{commt-tau}.
$t$ and $\eta$ commute. $t$, $\tau$
and $d$ satisfy the following
relations.
$$
\begin{array}{c}
t_{ij}d = d t_{ij}-(q^{-1}-q)\eta_{i5} \ttau_{5i} \\ \\
\ttau_{5j}d = d \ttau_ {5j}
\end{array}
$$
$\ttau$ and $\eta$ commute with each other, while finally
$$
\eta_{j5}d=q^{-1}d\eta_{j5}.
$$
\end{remark}

In analogy with the classical (non quantum) supersetting, we give
the following definition.

\begin{definition}
We  define the \textit{quantum big cell}  $\cO_q(U_{12})$
as the subring of $\cO(P_{lq})$
generated by $t$ and $\ttau$.
\end{definition}

We compute now the quantum commutation relations among
the generators of the quantum big cell $\cO_q(U_{12})$ , which is our
chiral Minkowski superspace, and see that the quantum big cell
admits a well
defined coaction of the quantum supergroup $\cO(P_{lq})$.

\begin{proposition} \label{commt-tau}
The quantum big cell superring $\cO_q(U_{12})$ has
the following  presentation:
$$
\cO_q(U_{12}) := \C_q \langle t_{ij}, \ttau_{5j} \rangle \, \big/ \, I_U,
\qquad 3 \leq i \leq 4, \, j=1,2
$$
where $I_U$ is the ideal generated by the relations:
$$
\begin{array}{c}
t_{i1}t_{i2}=q \, t_{i2}t_{i1}, \qquad
t_{3j}t_{4j}=q^{-1} \, t_{4j}t_{3j},
\qquad 1 \leq j \leq 2, \quad 3 \leq i \leq 4
\\ \\
t_{31}t_{42}=t_{42}t_{31}, \qquad
t_{32}t_{41}=t_{41}t_{32}+(q^{-1}-q)t_{42}t_{31},
\\ \\
\ttau_{51}\ttau_{52}=-q^{-1}\ttau_{52}\ttau_{51}, \qquad
t_{ij}\ttau_{5j}=q^{-1}\ttau_{5j}t_{ij}, \qquad 1 \leq j \leq 2 \\ \\
t_{i1}\ttau_{52}= \ttau_{52}t_{i1},
\qquad  t_{i2}\ttau_{51}=\ttau_{51} t_{i2}+(q^{-1}-q)t_{i1}\ttau_{52}.
\end{array}
$$
\end{proposition}

As in the classical setting we have the following proposition.

\begin{proposition}
The quantum big cell $\cO_q(U_{12})$ admits a coaction of
$\cO(P_{l,q})$ obtained by
restricting suitably the comultiplication in
 $\cO{(P_{l,q})}$.
In other words we have a well defined morphism:
$$
\begin{array}{ccc}
\tilde\Delta:\cO_q(U_{12}) & \lra &  \cO(P_{l,q}) \otimes \cO_q(U_{12})
\end{array}
$$
satisfying the coaction properties and give explicitly by:
(see \ref{coaction on big cell}),
\begin{align*}\tilde\Delta t_{ij}&=t_{ij}\otimes 1+y_{ia} S(x)_{bj}\otimes t_{ab}+y_{i}\eta_aS(x)_{bj}\otimes \ttau_{jb},&
\\ \\
\tilde\Delta\ttau_j&=(d\otimes 1)(\tau_a\otimes 1+\xi_b\otimes t_{ba}+1\otimes \ttau_a)(S(x)_{aj}\otimes 1)&.
\end{align*}
by choosing as before generators $x$, $y$, $t$, $d$, $\tau$, $\eta$, $\xi$
for $\cO(P_{l,q})$ and $t$, $\ttau$ for $\cO_q(U_{12})$ with
$d\tau=\ttau x$.

\medskip

Furthermore, this coaction goes down to a well defined coaction for
the quantization of the super Poincar\'{e} group (see remark \ref{qsuperpoinc}).

\end{proposition}

To compare with other deformations of the Minkowski space, we write here the even part of $\cO_q(U_{12}$ in terms of the more familiar generators

$$
t=x^\mu \sigma_\mu = \begin{pmatrix} x^0 + x^3 & x^1 - i x^2 \\ x^1 + i x^2 & x^0 + x^3 \end{pmatrix}.
$$

The commutation relations of the generators  $x^\mu$ are then \cite{cfln}
\begin{eqnarray*}
x^0 x^1 & = & \frac{2}{q^{-1}+q} x^1 x^0 + i \frac{q^{-1}-q}{q^{-1}+q} x^0 x^2, \\
x^0 x^2 & = & \frac{2}{q^{-1}+q} x^2 x^0 - i \frac{q^{-1}-q}{q^{-1}+q} x^0 x^1, \\
x^0 x^3 & = & x^3 x^0, \\
x^1 x^2 & = & \frac{i(q^{-1}+q)}{2} \left(- (x^0)^2 +(x^3)^2 + x^3 x^0 -x^0 x^3\right),\\
x^1 x^3 & = & \frac{2}{q^{-1}+q} x^3 x^1 - i \frac{q^{-1}-q}{q^{-1}+q} x^2 x^3, \\
x^2 x^3 & = & \frac{2}{q^{-1}+q} x^3 x^2 + i \frac{q^{-1}-q}{q^{-1}+q} x^1 x^3. \\
\end{eqnarray*}

\section{Chiral superfields  in Minkowski superspace} \label{chiral}

In this section we wish to motivate the importance of the
chiral conformal superspace and its quantum deformation in physics.
We introduce  chiral superfields in Minkowski superspace as they are used in physics. We start by introducing the complexified Minkowski space: the chiral superfields are a sub superalgebra of the coordinate superalgebra of Minkowski space. They can also be seen as the coordinate superalgebra of the chiral Minkowski superspace, which is complex.

\subsection{Definitions}

We consider the complexified Minkowski space $\C^4$. The $N=1$ {\sl scalar superfields} on the complexified Minkowski space are elements of the commutative superalgebra  \be\cO(\C^{4|4})\equiv C^\infty(\C^4)\otimes \Lambda[\theta^1,\theta^2,\bar\theta^1,\bar\theta^2],\label{superminkowski}\ee where $ \Lambda[\theta^1,\theta^2,\bar\theta^1,\bar\theta^2]$ is the Grassmann (or exterior) algebra generated by the odd variables $\theta^1,\theta^2,\bar\theta^1,\bar\theta^2$.

We will denote the coordinates (or generators) of the superspace
as
\begin{eqnarray*}&x^\mu, \qquad &\mu=0,1,2,3 \quad \hbox{(even
coordinates),}\\&\theta^\alpha,\bar \theta^{\dot\alpha},\qquad&
\alpha, \dot \alpha=1,2 \quad \hbox{(odd coordinates),}
\end{eqnarray*}
and a  superfield, in terms of its {\sl field components}, as
\begin{align*}\Psi(x, \theta,\bar \theta)=&\psi_0(x)+\psi_\alpha(x)\theta^\alpha+\psi'_{\dot\alpha}(x)\bar\theta^{\dot\alpha}+
\psi_{\alpha \beta} (x)\theta^{ \alpha} \theta^{ \beta}+\psi_{ \alpha\dot\beta} (x) \theta^{\alpha}\bar\theta^{\dot\beta}+\\&
\psi'_{\dot\alpha\dot\beta}(x)\bar\theta^{\dot\alpha}\bar\theta^{\dot\beta}+
\psi_{ \alpha \beta\dot\gamma}(x) \theta^{ \alpha} \theta^{ \beta}\bar\theta^{\dot\gamma}+
\psi'_{ \alpha\dot\beta\dot\gamma}(x) \theta^{ \alpha}\bar\theta^{\dot\beta}\bar\theta^{\dot\gamma}+
\psi_{ \alpha \beta\dot\gamma\dot\delta}(x) \theta^{ \alpha} \theta^{ \beta}\bar\theta^{\dot\gamma}
\bar\theta^{\dot\delta}.\end{align*}

\paragraph{Action of the Lorentz group SO(1,3).} There is an action of the double covering of the complexified Lorentz group,
$\rSpin(1,3)^c\approx\rSL(2,\C)\times \rSL(2,\C)$ over $\C^{4|4}$. The even coordinates  $x^\mu$ transform according to
the fundamental representation of $\rSO(1,3)$ ($V$),
$$x^\mu\mapsto \Lambda^\mu{}_\nu x^\nu,$$
while $\theta$ and $\bar \theta$ are Weyl spinors (or
half spinors). More precisely, the coordinates $\theta$ transform in one of the spinor representations, say
$S^+\approx(1/2,0)$ and $\bar \theta$ transform in the opposite chirality representation, $S^-\approx(0,1/2)$,
$$\theta^\alpha\mapsto {S}^\alpha{}_\beta\theta^\beta,\qquad \bar\theta^{\dot\alpha}\mapsto \tilde S^{\dot \alpha}{}_{\dot \beta}\theta^{\dot \beta}.$$

 The scalar superfields are invariant under the action of the Lorentz group,
$$\Psi(x,\theta,\bar\theta)= (R\Psi)(\Lambda^{-1} x, S^{-1}\theta,\tilde S^{-1}\bar\theta),$$ where $R\Psi$ is the superfield obtained by transforming the
field components
$$R\psi_0(x)=\psi_0(x),\quad R\psi_\alpha(x)=S_\alpha{}^\beta\psi_\beta(x), \;\; \dots$$

The hermitian matrices
$$\sigma^0=\begin{pmatrix}1&0\\0&1\end{pmatrix},\quad
\sigma^1=\begin{pmatrix}0&1\\1&0\end{pmatrix},\quad\sigma^2=\begin{pmatrix}0&-\ri\\\ri&0\end{pmatrix},\quad
 \sigma^3=\begin{pmatrix}1&0\\0&1\end{pmatrix},$$ define a
  $\rSpin(1,3)$-morphism
$$\begin{CD}S^+\otimes S^-@>>>V\\
s^\alpha \otimes  t^{\dot\alpha} @>>>
s^\alpha\sigma^\mu_{\alpha\dot\alpha}t^{\dot\alpha}.
\end{CD}$$

 \paragraph{Derivations.}
A {\sl left derivation} of degree $m=0,1$ of  a super algebra $\cA$ is a
linear map $D^L:\cA\mapsto \cA$ such that
$$
D^L(\Psi\cdot\Phi)=D^L(\Psi)\cdot \Phi +(-1)^{mp_\Psi}\Psi\cdot D^L(\Phi).
$$
Graded left derivations span a $\Z_2$-graded vector space (or {\sl supervector space}).

In general, linear maps over a supervector space are also  a $\Z_2$-graded vector space. A map has degree 0 if it preserves the
parity and degree 1 if it changes the parity. For the case of derivations of a commutative superalgebra,
 an even derivation has degree 0 as a linear map and an odd derivation
has degree 1 as a linear map.

In the same way one defines  {\sl right derivations},
$$
D^R(\Psi\cdot \Phi)=(-1)^{mp_\Phi}D^R(\Psi)\cdot \Phi +\Psi\cdot D^R(\Phi).
$$
Notice that  derivations of degree zero are both, right and left derivations.
Moreover, given a left derivation $D^L$ of degree $m$  one can define a
right derivation $D^R$ also  of degree $m$
 in the following way
\be D^R \Psi=
(-1)^{m(p_\Psi+1)}D^L\Psi.\label{rl}\ee

Let us now focus on the commutative superalgebra $\cO(\C^{4|4})$. We define the standard left derivations
\begin{align*}\partial_\alpha^L\Psi&=\psi_\alpha+
2\psi_{\alpha \beta} \theta^{ \beta}+\psi_{ \alpha\dot\beta} \bar\theta^{\dot\beta}+
2\psi_{ \alpha \beta\dot\gamma}  \theta^{ \beta}\bar\theta^{\dot\gamma}+
\psi'_{ \alpha\dot\beta\dot\gamma} \bar\theta^{\dot\beta}\bar\theta^{\dot\gamma}+
2\psi_{ \alpha \beta\dot\gamma\dot\delta}  \theta^{ \beta}\bar\theta^{\dot\gamma}
\bar\theta^{\dot\delta},\\
\partial^L_{\dot\alpha}\Psi&=\psi'_{\dot\alpha}-\psi_{ \beta\dot\alpha} \theta^{\beta}+
2\psi'_{\dot\alpha\dot\beta}\bar\theta^{\dot\beta}+
\psi_{ \gamma \beta\dot\alpha} \theta^{ \gamma}
\theta^{ \beta}
-2\psi'_{ \beta\dot\alpha\dot\gamma} \theta^{ \beta}\bar\theta^{\dot\gamma}+2
\psi_{ \gamma \beta\dot\alpha\dot\delta}
\theta^{ \gamma} \theta^{ \beta}
\bar\theta^{\dot\delta}.
\end{align*}
Also, using (\ref{rl})  one can define $\partial_\alpha^R, \partial^R_{\dot\alpha}$.

We consider now the odd left derivations
$$Q^L_\alpha=\partial^L_\alpha
-\ri\sigma^\mu_{\alpha\dot\alpha}\bar\theta^{\dot\alpha}\partial
_\mu,\qquad  \bar Q^L_{\dot \alpha}=-\partial^L_{
\dot \alpha}
+\ri\theta^{\alpha}\sigma^\mu_{\alpha\dot\alpha}\partial_\mu.$$ They satisfy the anticommutation rules
$$\{Q^L_\alpha, \bar Q^L_{\dot \alpha}\}=2\ri \sigma^\mu_{\alpha\dot\alpha}\partial_\mu, \qquad \{Q^L_\alpha, Q^L_{\beta}\}=\{\bar Q^L_{\dot
\alpha}, \bar Q^L_{\dot \beta}\}=0,$$ with $\partial_\mu=\partial/\partial x^\mu$. $Q^L$ and $\bar Q^L$ are the
{\sl supersymmetry charges} or {\sl supercharges}. Together with
$$P^\mu=-\ri \partial_\mu,$$ they form a Lie superalgebra,  the {\sl supertranslation
algebra}, which then acts on  the superspace $\C^{4|4}$.

Let us define another set of (left) derivations,
$$D^L_\alpha=\partial_\alpha
+\ri\sigma^\mu_{\alpha\dot\alpha}\bar\theta^{\dot\alpha}\partial_\mu,\qquad  \bar D^L_{\dot \alpha}=-\partial_{\dot \alpha}
-\ri\theta^{\alpha}\sigma^\mu_{\alpha\dot\alpha}\partial_\mu,$$ with anticommutation rules
$$\{D^L_\alpha, \bar D^L_{\dot \alpha}\}=-2\ri \sigma^\mu_{\alpha\dot\alpha}\partial_\mu, \qquad \{D^L_\alpha, D^L_{\beta}\}=\{\bar D^L_{\dot
\alpha}, \bar D^L_{\dot \beta}\}=0.$$ They also form a Lie
superalgebra, isomorphic to the supertranslation algebra. This can
be seen by taking
$$Q^L\rightarrow -D^L,\qquad \bar Q^L\longrightarrow \bar D^L.$$

It is easy to see that the supercharges anticommute with the
derivations $D^L$ and $\bar D^L$. For this reason, $D^L$ and $\bar
D^L$ are called {\sl supersymmetric covariant derivatives} or
simply {\sl covariant derivatives}, although they are not related
to any connection form.

We go now to the central definition.

\begin{definition}

A {\sl chiral superfield} is a superfield $\Phi$ such that
\begin{equation}\bar D^L_{\dot\alpha}\Phi=0.\label{chiralcons}\end{equation}
\end{definition}

Because of the anticommuting
properties of $D's$ and $Q's$, we have that
$$ \bar D^L_{\dot\alpha}\Phi=0\quad \Rightarrow\quad  \bar
D^L_{\dot\alpha}(Q^L_\beta\Phi)=0,\quad \bar D^L_{\dot\alpha}(\bar
Q^L_{\dot\beta}) \Phi=0.$$ This means that the supertranslation
algebra acts on the space of chiral superfields.

On the other hand, due to the derivation property,
$$\bar D^L_{\dot\alpha}(\Phi\Psi)=\bar D^L_{\dot\alpha}(\Phi)\Psi+(-1)^{p_\Phi}\Phi\bar
D^L_{\dot\alpha}(\Psi),$$ we have that the product of two chiral
superfields is again a chiral superfield.

\subsection{Shifted coordinates}
One can solve the constraint (\ref{chiralcons}). Notice that the
quantities
\be y^\mu=x^\mu
+\ri\theta^\alpha\sigma^\mu_{\alpha\dot\alpha}\theta^{\dot\alpha},\qquad
\theta^\alpha\label{shifted}\ee
satisfy
$$\bar D^L_{\dot\alpha}y^\mu=0,\qquad \bar
D^L_{\dot\alpha}\theta^\alpha=0.$$ Using the derivation
property, any superfield of the form
$$\Phi(y^\mu,\theta),\qquad \hbox{satisfies  }\quad \bar D^L_{\dot\alpha}\Phi=0$$ and so it  is a chiral
superfield. This is the general solution of (\ref{chiralcons}).

We can make the change of coordinates
$$x^\mu, \; \theta^\alpha,\; \bar\theta^{\dot\alpha}\;\longrightarrow \; y^\mu=x^\mu +i\theta^\alpha
\sigma^\mu_{\alpha\dot\alpha} \bar\theta^{\dot \alpha},\;
\theta^\alpha,\;\bar\theta^{\dot\alpha}.$$

A superfield may be expressed in both coordinate systems
$$\Phi(x, \theta, \bar\theta)=\Phi'(y, \theta, \bar\theta).$$ The
covariant derivatives and supersymmetry charges take the form
\begin{eqnarray*}
D^L_\alpha\Phi'  =\frac{\partial^L \Phi'}{\partial\theta^\alpha}
 +2i
\sigma^\mu_{\alpha\dot\alpha}\bar\theta^{\dot\alpha}
\frac{\partial^L \Phi' }{\partial y^\mu}\qquad \bar
D^L_{\dot\alpha}\Phi' =
-\frac{\partial^L \Phi'}{\partial\bar\theta^{\dot\alpha}},\\
  \bar Q^L_{\dot\alpha}\Phi'  =
-\frac{\partial^L \Phi'}{\partial\bar\theta^{ \dot\alpha}} +2i
\theta^\alpha\sigma^\mu_{\alpha\dot\alpha}\frac{\partial^L \Phi'
}{\partial y^\mu}\qquad  Q^L_\alpha\Phi' =
\frac{\partial^L\Phi'}{\partial\theta^\alpha}.
\end{eqnarray*}
In the new coordinate system the chirality condition is simply
$$\frac{\partial^L\Phi'}{\partial\bar\theta^{\dot\alpha}}=0,$$ so it
is similar to a holomorphicity condition on the $\theta$'s.

This shows that chiral scalar superfields are elements of the commutative superalgebra $\cO(\C^{4|2})=\C^\infty(\C^4)\otimes \Lambda[\theta^1,\theta^2]$.
In the previous sections
we realized this superspace as the big cell inside the chiral conformal superspace, which is the Grassmannian of $2|0$-subspaces of $\C^{4|1}$.

The complete (non chiral) conformal superspace is in fact the flag space of $2|0$-subspaces inside $2|1$-subspaces of  $\C^{4|1}$. On this supervariety one can put a reality condition, and the real Minkowski space is the big cell inside the flag.
It is instructive to compare Eq. (\ref{shifted}) with the incidence relation for the big cell of the flag manifold in Eq. (12) of Ref. \cite{flv}. We can then be convinced that the Grassmannian that we use to describe chiral superfields is inside the (complex) flag.

\bigskip

There are supersymmetric theories in physics (like Wess-Zumino models, or super Yang-Mills) that include in the formulation chiral superfields. In previous approaches  it has been difficult to formulate them on non commutative superspaces (with non trivial commutation relations of the odd coordinates). The reason was that the covariant derivatives are not anymore derivations of the noncommutative superspace, and the chiral superfields do not form
a superalgebra \cite{fel,flm}. Some proposals
to solve these problems include the partial (explicit) breaking of supersymmetry \cite{se,flm}. In our approach to quantization of superspace, the quantum chiral ring appears in a natural way, thus making possible the formulation of supersymmetric theories in non commutative superspaces. Also, the super variety and the supergroup acting on it become non commutative, the group law is not changed, so the physical symmetry principle remains intact. This is a virtue of the deformation based on quantum matrix groups.

\end{document}